\documentclass[12pt,
tightenlines,
floats,
showpacs,
nofootinbib,
amsmath,
amssymb,
aps,
superscriptaddress,
prd, showkeys]{revtex4-1}

\usepackage{graphicx}
\usepackage{amsmath,amssymb}
\usepackage{hyperref}
\usepackage{mathrsfs}

\newtheorem{Lemma}{Lemma}[section]
\newtheorem{proof}{Proof}[section]
\newtheorem{Proposition}{Proposition}[section]

\begin{document}
	
\title{Comments on Mathematical Aspects of the Bir\'o--Néda Model}

	%\vspace{6pt} 

\author{Ilda Inácio}
\email{ilda@ubi.pt}
\affiliation{Centro de Matemática e Aplicações (CMA-UBI), Universidade da Beira Interior, 6201-001, Covilhã, Portugal}
\author{Jos\'e Velhinho}
\email{jvelhi@ubi.pt}
\affiliation{Faculdade de Ci\^encias and FibEnTech-UBI, Universidade da Beira Interior, R. Marqu\^es D'\'Avila e Bolama, 6201-001 Covilh\~a, Portugal}
 
%\vspace{6pt} 

\begin{abstract}
{\bf Abstract:} We address two mathematical aspects of the  Bir\'o--Néda dynamical model, recently applied in the statistical analysis of several and varied complex phenomena. First, we show that a given implicit assumption ceases to be valid  outside the most simple and common cases, and we analyze  the consequences thereof, in what the formulation of the model  and probability  conservation  is concerned. Second, we revisit the transient behavior in the case of a constant reset rate and a constant or linear growth rate, improving on a previous analysis by including more general initial conditions.
\end{abstract}

\keywords{probability conservation; stationary distribution; transient behavior}

\maketitle

\section{Introduction}
\label{sec:Intro}

In recent years, Biró, Néda, and collaborators have successfully applied a simple dynamical model to address the behavior of complex systems of a varied nature, including high energy physics systems, ecology, population distribution,
scientific citations, and  social media popularity 
\cite{1,2, face,8}. In addition, particularly promising are more recent applications to income and wealth distributions 
\cite{4,5,6}. These latter applications contribute to an ongoing effort to model social and economic phenomena by means of simple mathematical models, often inspired in statistical physics and stochastic  processes (see e.g., \cite{w1,w2,w3}). Going beyond the standard master equation for diffusion in a discrete system, the authors intend to explore unidirectional growth processes in order to model the dynamics of complex systems, including open ones. These  processes explicitly break the detailed balance condition, and the unidirectional growth is  supplemented with a mechanism designed to ensure the existence of interesting stationary distributions 
\cite{1,2}.  

Specifically, the model concerns a probability distribution 
$\{P_n,n\geq 0\}$ evolving in time according to the following differential equations:
%\begin{eqnarray}
\begin{equation}
\label{0}
\dot{P_n}=-(\gamma_n + \mu_n) P_n + \mu_{n-1}P_{n-1}, \qquad n\geq 1,\\
%\end{eqnarray}
\end{equation}
\begin{equation}
\label{1}
%\label{sum}
\dot{P}_0 =\sum_{n=0}^{\infty}\gamma_n P_n -  (\gamma_0 +\mu_0) P_0.  
\end{equation}

 Within this class of models, each particular case is parametrized by the sequences   $\{\mu_n\}$ and 
$\{\gamma_n\}$.
One can recognize in (\ref{0}) the differential equations for the probability of each state in a  continuous-time Markov chain \cite{VK,R}.   
%  $\mathbb{N}_0$. %
In  particular, $\{\mu_n\}$ characterizes the dynamics  of a pure birth process \cite{VK,R}.   
However, the one-step transitions are supplemented by nonlocal transitions to the state $n=0$, modifying the 
dynamics in order to allow for nontrivial stationary solutions.

In fact, the time-independent version of (\ref{0}), i.e., with 
$\dot{P_n}=0$, for all $n\geq 1$, 
%\begin{equation}
%\label{2}
%0=-\lambda_n Q_n + \mu_{n-1}Q_{n-1}, \qquad \forall n\geq 1,
%\end{equation}
 can be used by iteration as a generator for probability distributions, starting from some $Q_0>0$. 
Several interesting distributions have been considered and discussed in applications, for various $\{\gamma_n\}$ 
and $\{\mu_n\}$. 
In particular, for constant $\gamma_n$, the standard exponential distribution $Q_n=Q_0e^{-\beta n}$ emerges for the case of constant $\mu_n$, whereas a linear growth 
$\mu_n=\sigma(n+b)$ produces the so-called Waring distribution \cite{War}.
Other distributions of interest are generated by other sequences, with constant and nonconstant $\gamma_n$ and a variety of behaviors for  $\mu_n$, including power-law and exponential growth (see Refs. \cite{1,2}, and  also 
\cite{cor} for the emergence of distributions in a related context). 
%In particular, the general master equation for a discrete system corresponding to a continuous-time Markov process, 
%\begin{equation}
%\label{00}
%\dot{P_n}=\sum_m [\omega_{mn}P_m - \omega_{nm}P_n] , 
%\end{equation}

%{The Biró-Néda model \cite{1,2} therefore consists of the  system (\ref{0},\ref{1}),
%providing a dynamical framework for the evolution of complex systems, leading in particular to the distributions determined by system (\ref{2}) as stationary states.} 

The purpose of the current article is twofold. In Section \ref{trans}, we address the behavior of the transient, 
time-dependent solutions of the Biró--Néda model for the two most simple and common cases, namely with constant 
$\gamma_n$ and constant or linear $\mu_n$. In this respect, we develop the analysis  first put forward in Ref. 
\cite{3},  hopefully contributing to the effort of exploring the behavior of these models, and in particular to improve the understanding of the approach to stationarity.

In the first part of the article, we point out and address technical challenges faced by the Biró--Néda model, outside the  simple cases mentioned above. First, we will argue in Section~\ref{stat}   that the  current formulation of  the model requires modification, in order to accommodate the cases such that  the sequence of values
$\mu_n/\gamma_n$ grow faster than $n$. In fact, for those cases, we will show that the model does not actually admit 
(nontrivial) stationary solutions. The above-mentioned iteration process is of course present, however the only stationary solution compatible with 
(\ref{1}) gives $Q_0=0$. On the technical side, we  disprove a claim made in Ref. \cite{1}, showing that it fails for those fast sequences $\mu_n/\gamma_n$.

As a second contribution, we discuss the solution to the above  issue.  We will see in Section \ref{form} that 
there is a natural and simple modification of Equation (\ref{1}) compatible with the  stationary solutions, however it comes at a cost. In fact,
we will argue  that probability is only conserved close to stationarity, which in this case means that it can only be conserved  starting from  initial  distributions that are quite different from   typical initial states. 
%We will also argue that this consequence is unavoidable.
This consequence seems to be  unavoidable, since the 
%time independent distribution $\{Q_n\}$ generated from $Q_0>0$ by means of (\ref{2}) is such that   $\lim_{n\to\infty} \mu_n Q_n\not = 0$, and therefore 
stationary solution introduces a nonzero boundary term in the equation for probability conservation.

It should  again be said that no such issues affect the most simple and most common applications discussed by Biró, Néda, and collaborators, which most frequently involve simple situations for which the above-mentioned claim made in Ref.  \cite{1} is actually valid. Nevertheless, the departure from those simple cases is relevant, and  our remarks shed  light on the limits of applicability of the model, hopefully contributing to improve the mathematical formulation and the mathematical status of the model itself. 

%The mathematical-physics consequences in terms of viability of the models  are discussed in .

\section{Stationary Solution}
\label{stat}
%Focus on case $\lambda_n-\mu_n=cte$. Leave geberal case for comments on section discussion

In this section we will disprove a general claim made in Ref. \cite{1} and show that, for the cases where that claim is not valid, the current formulation of the Biró--Néda model \cite{1,2} does not admit a nontrivial stationary solution.

The
Biró--Néda model  consists of the system of Equations (\ref{0}) and (\ref{1}) for the  evolution of a probability  distribution 
$\{P_n(t),n\geq 0\}$, 
 where the dot stands for time derivative, and time $t$ takes values in 
$[0,\infty[$. Following Refs. \cite{1,2}, we introduce  parameters $\lambda_n$:
\begin{equation}
\lambda_n=\mu_n+\gamma_n, \qquad n\geq 0,
\end{equation}
where all values $\mu_n$  and $\gamma_n$ are positive. The sequence of values  $\gamma_n$ is referred to by the authors of 
Refs. \cite{1,2} as the loss rate, or reset rate, and $\mu_n$ as the growth rate.

Let us then suppose that a stationary solution of the model exists, i.e., such that $\dot{P}_n=0$, 
for all $t$, for all $n\geq 0$. Let $P_n(t)=Q_n$ be that solution,  where
 $\{Q_n, n\geq 0\}$ is some sequence of real numbers satisfying:
\begin{equation}
\label{2}
0=-\lambda_n Q_n + \mu_{n-1}Q_{n-1}, \qquad \forall n\geq 1,
\end{equation}
and a corresponding condition coming from 
 Equation (\ref{1}), which is:
\begin{equation}
\label{g5}
\sum_{n=0}^{\infty}\gamma_n Q_n=Q_0\lambda_0.
\end{equation}

On the other hand, one can easily obtain the expression for $Q_n$, for all $n\geq 1$, in terms of $Q_0$, 
by a standard iteration typical of these systems \cite {R}. Starting from $Q_0$, the  iteration of system (\ref{2}) gives:
\begin{equation}\label{22}
 Q_n = Q_0 \frac{\mu_0\mu_1\ldots\mu_{n-1}}{\lambda_1\lambda_2\ldots\lambda_{n}}, \qquad n\geq 1.
\end{equation}

Following Ref.  \cite{1}, let us rewrite (\ref{22}) in the form   (valid also for $n=0$):
\begin{equation}
\label{g3}
Q_n=\frac{Q_0\lambda_0}{\gamma_n}r_n\prod_{k=0}^{n}\frac{1}{1+r_k}, \quad \forall n\geq 0, \quad \mbox{where}
\quad
 r_k:=\frac{\gamma_k}{\mu_k}.
\end{equation}

Again as in Ref. \cite{1}, we introduce also a special notation $S_0$ for the following infinite sum:
 \begin{equation}
\label{S0}
S_0=\sum_{n=0}^{\infty}Z_n,
\end{equation}
where we have defined:
\begin{equation}
\label{3}
Z_n:=r_n\prod_{k=0}^{n}\frac{1}{1+r_k}.
\end{equation}

Note that  $S_0$ depends exclusively on the sequence $\{r_n=\gamma_n/\mu_n\}$.

It follows immediately from (\ref{g3}) that:
\begin{equation}
\label{g4}
\sum_{n=0}^{\infty}\gamma_n Q_n=Q_0\lambda_0 S_0.
\end{equation}

Thus, the authors of Refs. \cite{1,2} have two distinct expressions for $\sum^{\infty}\gamma_n Q_n$: Equation~(\ref{g4}) above and Equation (\ref{g5}),  obtained from the evolution equation for ${P}_0$ (\ref{1}). To reconcile these two expressions, it is claimed in Ref. \cite{1} that $S_0$ always takes the value $S_0=1$, for all sequences $\{r_n\}$.
%There are several issues with these arguments. First, 
We  show next that this is not the case. In fact, and although $S_0$ is indeed equal to one for the most common and simple type of sequences $\{\mu_n\}$, $\{\gamma_n\}$ considered in Refs. \cite{1,2}, it turns out that $S_0$ is strictly smaller than one for  more general sequences $\{r_n\}$. 
%Second, the two expressions should not need to be reconciled in the first place, since conservation of probability leads indeed to equation (\ref{g4}).

%and that result is moreover crucially used in checking the consistency of the formulation. 

%Differential equations for the time evolution of the probability distribution $P_n(t)$, with  $\sum_{n=0}^{\infty} P_n =1$, $\forall t$ (the dot stands for time-derivative):

\begin{Lemma}\label{gb1}
\begin{equation}
\sum_{n=0}^N Z_n = 1 - Z_N/r_N.
\end{equation}
\end{Lemma}
\begin{proof}
Follows from the sum of Equation (\ref{2}), or alternatively, directly from (\ref{3}), by working out 
the expression for $\sum^N Z_n - 1$. $\square$
\end{proof}

\begin{Proposition}\label{s}
The infinite sum $S_0$ exists for all (positive) sequences $\{r_n\}$, with 
$S_0\leq 1$. The equality $S_0= 1$
 is achieved if and only if the sum  $\sum_{n=0}^N \ln(1+r_n)$ diverges, when $N\to\infty$.
\end{Proposition}
\begin{proof}
From Lemma \ref{gb1}, it follows that:
\begin{equation}
\label{so}
S_0 = 1 - \lim_{N\to\infty} \prod_{n=0}^{N}\frac{1}{1+r_n}.
\end{equation}

The latter limit always exists, since the sequence is clearly decreasing and bounded from below.
On the other hand, applying logarithm to the product, it follows that the limit of $\prod^{N}\frac{1}{1+r_n}$ is zero if and only if the sum $\sum^N \ln(1+r_n)$ diverges. $\square$
\end{proof}

{Note that $S_0 = 1$
is ensured for all sequences $\{r_n\}$ that do not tend to zero, and also for  sequences $\{r_n\}$ that behave like 
$1/n$, for large $n$. However, $S_0$ fails to be equal to one for  sequences $\{r_n\}$ that behave, for  large $n$,  as 
$r_n\sim n^{-1-\epsilon}$, with $\epsilon >0$, since the infinite sum $\sum^{\infty} \ln(1+r_n)$ behaves in that case 
(apart from some finite term)  as 
$\sum^{\infty} n^{-1-\epsilon}$, which converges.  Thus, $S_0 = 1$ is indeed satisfied for the simplest cases considered in Refs.~\cite{1,2,3,4,5,6}}, and in particular for constant  sequences $\{\gamma_n\}$, with $\{\mu_n\}$ being also constant or linearly growing with $n$, but it fails even for  constant $\{\gamma_n\}$ sequences if 
$\{\mu_n\}$ grows like $n^{1+\epsilon}$.

If $S_0$ is not equal to one, then the Biró--Néda model simply does not admit a nontrivial stationary solution. In fact, the only solution to the joint Equations (\ref{g5}) and (\ref{g4}) is in that case $Q_0=0$, since $\lambda_0\not=0$. This of course implies $Q_n=0$, for all $n$. 

The source of the mismatch between expressions  (\ref{g5}) and (\ref{g4}) present in Refs. \cite{1,2} becomes clear when considering the derivative of the total probability. 
%(\ref{sum}) does not actually ensure the conservation of the probability, precisely because $S_0$ is not necessarily equal to one. 
Let us then apply the time derivative to the infinite sum $\sum^{\infty}P_n$, taking into account 
(\ref{0}) (and assuming that the derivative and limit of the sum commute). We  obtain: 
\begin{equation}
\label{sum1}
\frac{d}{dt}\sum_{n=0}^{\infty}P_n=\dot{P}_0+\lim_{N\to\infty}\sum_{n=1}^{N}\dot{P}_n=\dot{P}_0 -\lim_{N\to\infty}\sum_{n=1}^{N}\lambda_n P_n + \lim_{N\to\infty}\sum_{n=1}^{N}\mu_{n-1} 
P_{n-1}. 
\end{equation}

A rearrangement of the finite sums then gives:
\begin{equation}
\label{sum2}
\frac{d}{dt}\sum_{n=0}^{\infty}P_n=\dot{P}_0 -\lim_{N\to\infty}\sum_{n=0}^{N}\gamma_n P_n + (\mu_0 + \gamma_0) P_0 -
\lim_{N\to\infty}\mu_{N}P_N.
\end{equation}

Typically, one would like to neglect the last term on the right-hand side, and the requirement of probability conservation would lead to Equation (\ref{1}). However, this is incompatible with the possibility of a nontrivial stationary solution, if
 %$\lim_{N\to\infty}\mu_{N}P_N$  remains. This extra term cannot be assumed to be zero if 
$S_0\not =1$. In fact, for a distribution $\{Q_n\}$ generated from $Q_0$ by means of the system 
(\ref{2}), it follows from  (\ref{g3}) and (\ref{so}) precisely that $\lim_{N\to\infty}\mu_{N}Q_N=
\lambda_0Q_0(1-S_0)$, and thus $\lim_{N\to\infty}\mu_{N}Q_N$ is zero for a nontrivial stationary solution if and only if $S_0=1$.
%%
%In particular, restricting   Equation (\ref{sum2})  to stationary solutions  gives
%\begin{equation}
%\label{sum3}
%0 =\lim_{N\to\infty}\sum_{n=0}^{N}\gamma_n Q_n - \lambda_0 Q_0 + 
%\lim_{N\to\infty}\mu_{N}Q_N,
%\end{equation}
%
%which leads exactly to Equation (\ref{g4}), and not to  (\ref{g5}).
% as claimed in Refs. \cite{1,2}. 

Equation (\ref{1}) therefore does not admit (nontrivial) stationary solutions when $S_0\not =1$. In order to include stationary solutions in the model, a modified version of Equation (\ref{1}) is required, which we discuss in the next section.

To conclude the present section, let us remark that, regardless of the differential equation for $P_0$, the (nontrivial) stationary solution 
to be defined by (\ref{g3}) is not necessarily normalizable for arbitrary sequences $\{\gamma_n\}$, $\{\mu_n\}$. In fact, taking $\gamma_n=\mu_n=e^{-n}$ leads to $Q_n\propto (e/2)^n$, which is clearly not summable. 
%On the other hand, it is always possible  to normalize the sequence $Q_n$ for constant $\gamma_n=\gamma$, or to modify Equation (\ref{1}) accordingly (e.g., replacing $\gamma$ by $\gamma/ S_0$), but there remains the issue  of the dynamical formulation and its compatibility with the conservation of probability, which we address next. 
On the other hand, considering in particular situations with $S_0\not =1$, normalizability is definitely ensured for the cases of fast growth  rates described  e.g., in 
Refs. \cite{1,2}.  In fact, it follows from $S_0\not =1$ that $Q_n$ behaves asymptoticaly as $1/\mu_n$, and  is therefore normalizable for all sequences $\mu_n$ growing faster than $n$

\section{Dynamics and Probability Conservation}
\label{form}
We will now  obtain and discuss the generalization of Equation (\ref{1}) to the cases where $S_0\not=1$, allowing for nontrivial stationary solutions.

\subsection{General Case}
\label{general}
The simple  way out of the conflict between Equations (\ref{g5}) and (\ref{g4}) is to  include  the boundary term $\lim_{n\to\infty}\mu_{n}Q_n= Q_0\lambda_0(1-S_0)$ as an extra nonhomogeneous contribution  in Equation (\ref{1}), or alternatively to redefine the parameters 
$\gamma_n$ appearing in (\ref{1}) (and only those, Equations (\ref{0}) are left unchanged). 
%which vanishes if $S_0=1$, in the form: 
Let us start with a free parameter $R>0$, define: 
\begin{equation}\label{r}
\tilde{\gamma}_n=\gamma_n+R\lambda_0(1-S_0), 
\end{equation}
and adopt the following equation for 
$\dot{P}_0$:
\begin{equation}
\label{sum17I}
\dot{P}_0 =\sum_{n=0}^{\infty}\tilde{\gamma}_n P_n - \lambda_0  P_0.
\end{equation}

For stationary solutions we now obtain from (\ref{sum17I}):
\begin{equation}
\label{sum17Is}
 \sum_{n=0}^{\infty}\gamma_n Q_n = \lambda_0  Q_0 - R\lambda_0(1-S_0)\sum_{n=0}^{\infty} Q_n,
\end{equation}
and a solution compatible with (\ref{g4}) emerges, with $Q_0=R$ and $\sum_{n=0}^{\infty} Q_n =1$.
%\begin{equation}
%\label{16}
%\sum_{n=0}^{\infty} Q_n =1.
%\end{equation}
The value of $R$, and therefore of $Q_0$, is thus fixed by  normalization.
 %\begin{equation}\label{rr}
%R^{-1}=\lambda_0 \sum_{n=0}^{\infty}
%\frac{1}{\mu_n}\prod_{k=0}^{n}\frac{1}{1+r_k},
%\end{equation}
%taking into account (\ref{g3}).
%whenever possible. 
%As  mentioned above, normalizability is not guaranteed for arbitrary sequences 
%$\{\gamma_n\}$, $\{\mu_n\}$, but it is definitely ensured for  situations of interest, described  e.g., in 
%Refs. \cite{1,2}, including fast growth  rates.

There is however a price to be paid,  in terms of probability conservation, for including stationary solutions in the model, with $S_0\not =1$, precisely due to the nonzero boundary term. To see this, let us go back to the general Equation (\ref{sum2}) and introduce Equation (\ref{sum17I}), to~obtain:
\begin{equation}
\label{sum8g}
\frac{d}{dt}\sum_{n=0}^{\infty}P_n= Q_0\lambda_0(1-S_0) \sum_{n=0}^{\infty}P_n-
\lim_{N\to\infty}\mu_{N}P_N.
\end{equation}

(Alternatively, adding the boundary term simply as a nonhomogeneous contribution in (\ref{1}) would have the only effect of removing the sum $\sum^{\infty}P_n$ from the right-hand side.)
% The price to pay is that the situation concerning  probability conservation does not seem favorable to practical applications of the model. In fact, 

Equation
(\ref{sum8g}) indicates that probability can only be conserved for initial probability distributions
$\{P_n(0)\}$ such that $\lim_{n\to\infty}\mu_{n}P_n(0)=Q_0\lambda_0(1-S_0)$. When the boundary term contribution is absent, the  typical initial distributions, such that $P_n(0)=0$ for all $n$ greater than a certain  $n_0$, would ensure a null right-hand side in (\ref{sum8g}). This is no longer the case  if
$S_0\not =1$. According to  (\ref{sum8g}), probability is only conserved  close to stationarity, e.g., for distributions with the same asymptotic behavior as the stationary solution. The usefulness of this model therefore seems to be restricted in  cases where $S_0\not =1$, since the  states with $P_n(0)=0$ for all 
$n>n_0$ are  typical initial states.

\subsection{Constant Reset Rate}
\label{form1}
As preparation for the next section, and also for its own intrinsic interest, we now consider the special case of constant $\{\gamma_n\}$ sequences, i.e., with $\gamma_n=\gamma$, for all $n$. 

Considering Equation (\ref{sum17I}), it is certainly tempting to simply impose  $\sum^{\infty} P_n=1$ for all $t$ and leave the remaining unchanged. However, since probability conservation is itself an issue 
(for  $S_0\not =1$),  it is safer to start afresh with a nonhomogeneous equation for $\dot{P}_0$ of the form:
\begin{equation}
\label{n1}
\dot{P}_0 =\gamma - \lambda_0  P_0 + K,
\end{equation}
where $K$ is some constant, and see where it leads. Equation (\ref{n1}) now gives $Q_0$ in terms of $K$, and compatibility with  (\ref{g4}) and normalization of 
$\{Q_n\}$ fixes $K$ as $K=\gamma/S_0 -\gamma$. 
We therefore  obtain:
\begin{equation}
\label{sum7}
\dot{P}_0 = - \lambda_0  P_0 + \frac{\gamma}{S_0},
\end{equation}
whereas (\ref{sum8g}) is replaced with:
\begin{equation}
\label{sum8}
\frac{d}{dt}\sum_{n=0}^{\infty}P_n= \gamma\left(1-\sum_{n=0}^{\infty}P_n\right) + Q_0\lambda_0(1-S_0)-
\lim_{N\to\infty}\mu_{N}P_N.
\end{equation}

Thus, with  $\sum^{\infty} P_n=1$, both Equations (\ref{sum7}) and (\ref{sum8}) coincide with those obtained from 
(\ref{sum17I})--(\ref{sum8g}).

It is worthwhile mentioning the following. As we have seen, for $S_0\not=1$, Equation~(\ref{1}) is fully incompatible with (nontrivial) stationary solutions. This is not the case with its reduced counterpart, e.g., an equation of the form  
$\dot{P}_0 = - \lambda_0  P_0 + {\gamma}$, since (\ref{g5}) would be replaced by $\gamma = Q_0\lambda_0$, which is of a different nature, and not fully incompatible with~(\ref{g4}). In fact, a stationary solution exists, with 
$\sum_{n=0}^{\infty} Q_n=S_0$. However, this does not mean that an unnormalized stationary solution for   Equation 
(\ref{1}) exists, as it does not. Note also that the reduction from Equation (\ref{1}) to its reduced counterpart would be inconsistent if $S_0\not =1$, since it would be done assuming $\sum_{n=0}^{\infty} P_n=1$, whereas the reduced equation gives $\sum_{n=0}^{\infty} Q_n =S_0$.

Returning to Equation (\ref{sum8}), the situation regarding probability conservation remains as discussed in Section 
\ref{general} above. Examples leading to $S_0\not=1$ are 
$\mu_n=n^s+1$, with $s>1$. In this case, 
any initial probability distribution decaying faster than $1/n^s$ fails to produce a null right-hand side   in 
(\ref{sum8}).   For bounded or linearly growing sequences $\{\mu_n\}$, on the other hand,  the model formed by system (\ref{0}), with $\gamma_n=\gamma$, and Equation (\ref{sum7}) (with $S_0=1$) is fully consistent to describe the time evolution of a probability 
distribution $\{P_n\}$.

\section{Behavior of the Transient Solution}
\label{trans}
In this section, we depart from foundational and formulation issues, and  address the behavior of the time-dependent solutions of the Biró--Néda model for two special cases. In particular, we consider sequences 
$\{\gamma_n\}$ that are constant, and  sequences $\{\mu_n\}$ that are either constant or growing linearly with $n$. In both cases the value of $S_0$ (\ref{S0}) is equal to one, no boundary issues such as those discussed in Sections 
\ref{stat} and \ref{form} exist, and the model is  fully consistent, as we have just seen.

This issue was originally addressed in Ref. \cite{3}. In that work, the general solutions of the differential equations were written down, for the two considered cases, and an attempt was made to study the behavior of the transient, i.e., 
time-dependent, solution. Particular attention was paid to the  monotonicity, or lack thereof, of the time-dependent solution when approaching  the stationary solution. Although the general solutions presented in Ref.~\cite{3} are absolutely correct, the analysis of their behavior was restricted to a very special case, embodied by very particular initial conditions. We show here that the simple behavior of the solutions discussed in Ref.  \cite{3} is far from generic.
%substantially improving on the analysis put forward in 

Let us then fix $\gamma_n=\gamma$, for all $n\geq 0$, and let  $\{P_n\}$ be a sequence of variables satisfying Equations (\ref{0}) and (\ref{sum7}) with $S_0=1$, for sequences $\{\mu_n\}$ that grow, at most, linearly with $n$. 
%It follows from the discussion in Section \ref{form1} that the value of $S_0$ (\ref{S0}) is equal to one, that the total probability is conserved for sequences such that $\lim_{n\to\infty}n P_n=0$, and that the corresponding stationary distribution $Q_n$  (\ref{g3}) is normalized, with $Q_0=\gamma/\lambda_0$.

Let us write, for arbitrary $n$:
\begin{equation}
\label{t0}
P_n(t)=\Delta_n(t)+Q_n.
\end{equation}

Since, for each fixed $n$, the set $\{Q_0,\cdots,Q_n\}$ is a solution of the nonhomogeneous system of differential equations determined by (\ref{0}) and (\ref{sum7}) for the variables $\{P_0,\cdots,P_n\}$, it follows from general arguments that
the variables $\Delta_n$ are solutions of the corresponding homogeneous system, i.e., 
\begin{eqnarray}
\label{1z}
\dot\Delta_0=-\lambda_0 \Delta_0,  \nonumber \\ 
\dot\Delta_n=-\lambda_n \Delta_n + \mu_{n-1}\Delta_{n-1}, \ n\geq 1.
\end{eqnarray}

%Since both $\{P_n\}$ and $\{Q_n\}$ are supposed to be normalized, it follows that $\{\Delta_n(t)\}$ is required to satisfy the condition (\ref{p0}).

Given that all  values $\lambda_n$, $n\geq 0$, are positive, the functions $\Delta_n$ tend to zero as time tends to infinity, with $P_n$ approaching its stationary value $Q_n$. The detailed behavior of the functions $\Delta_n$ was analyzed in Ref. \cite{3} only for the particular initial conditions $P_0(0)=1$, $P_n(0)=0$, for all $n> 0$. The authors observed only one stationary point for each function $\Delta_n$ and concluded (analytically in the case of constant $\mu_n$ and based on numerical evidence in the other case) that the functions $\Delta_n$ are monotonic after a critical time $t_c$ that depends on $n$. We show next that the behavior  of the functions $\Delta_n$ is much more involved, in general. In fact, depending on the initial conditions, each function $\Delta_n$ can have up to $n$ stationary points, 
with the position of the last one being arbitrary. Thus, while it is true  that all functions  $\Delta_n$ are eventually monotonic, the time after which monotonicity is achieved can be arbitrarily large, for any given $n$.

\subsection{Constant Growth Rate}
\label{cte mu}
Let us then fix 
$\mu_n=\mu$ for all $n\geq 0$, besides $\gamma_n=\gamma$. As shown in Ref. \cite{3}, the solutions $\Delta_n$ to the homogeneous system determined by (\ref{1z}) can be written in terms of the initial values $\Delta_j(0)$, $j=0,\cdots, n$ in the form:
\begin{equation}
\label{m1}
\Delta_n(t)=e^{-\lambda t}f_n(t),\qquad \lambda=\gamma+\mu,
\end{equation}
where $f_n$ is a degree-$n$ polynomial:
\begin{equation}
\label{m2}
f_n(t)=\sum_{k=0}^{n} \Delta_{n-k}(0) \frac{(\mu t)^k}{k!}.
\end{equation}
\label{m3}

The derivative of $\Delta_n$ is easily obtained:
\begin{equation}
\dot{\Delta}_n=e^{-\lambda t}(\dot{f}_n-\lambda f_n).
\end{equation}

Thus, the stationary points of $\Delta_n$ coincide with the positive roots of the function $\dot{f}_n-\lambda f_n$, which is again
a polynomial of degree $n$. In particular:
\begin{equation}
\label{m2z}
\dot{f}_n-\lambda f_n=\sum_{k=0}^{n-1} \bigl[\mu\Delta_{n-k-1}(0)-\lambda\Delta_{n-k}(0) \bigr]\frac{x^k}{k!} -
\lambda\Delta_{0}(0)\frac{x^n}{n!}, \qquad x=\mu t.
\end{equation}

Thus, each $\Delta_n$ has at most $n$ stationary points, and  is therefore monotonic after some point.
However, for any given pair $\mu$, $\lambda$, one can find initial conditions such that, 
for any given $n$, the polynomial
$\dot{f}_n-\lambda f_n$ is fairly arbitrary, concerning both the number of positive roots and the position of its largest root. In fact, the following applies:
\begin{Proposition}\label{cm}
Let $\mu$, $\lambda$, and $n$ be given, and let $t_n$ denote the largest (positive) root of  $\dot{f}_n-\lambda f_n$.
Then, for any $M>1$, one can find initial values $\Delta_k(0)$, $\forall k\geq 0$, such that $\dot{f}_n-\lambda f_n$ has $n$ positive roots and $t_n>M$.
\end{Proposition}
\begin{proof}
Let $p_n(x)=\sum_{k=0}^n a_k x^n$ denote an arbitrary degree-$n$ polynomial. The equation 
$\dot{f}_n-\lambda f_n=p_n$ gives us a system of linear equations for the set $\Delta_0(0),\cdots, \Delta_n(0)$,
in terms of the set $a_0,\cdots,a_n$, which can always be solved.
Now choose $p_n$ such that it possesses $n$ positive roots, which is certainly possible (choose any polynomial with $n$ real roots and apply a translation). If the largest root of $p_n$ is not greater than the prescribed value $\mu M$, replace 
$p_n(x)$ with $p_n(x/\Lambda)$, for appropriate $\Lambda >1$. The choice $\dot{f}_n-\lambda f_n=p_n$ thus gives us $n$ positive roots and an arbitrarily large value of $t_n$. To prove the proposition, it remains to show that the initial values $\Delta_k(0)$, $k\leq n$,  can be chosen in a way that is compatible with a probability  distribution, namely defined by 
$P_k=\Delta_k +Q_k$. It is thus necessary that:
\begin{equation}\label{norm}
-Q_k\leq\Delta_k(0)\leq 1-Q_k,
\end{equation}
and 
\begin{equation}\label{norm2}
\sum_{k=0}^n\bigl(\Delta_k(0)+Q_n\bigr)\leq 1,
\end{equation}
which follow from $0\leq P_k\leq 1$, $\forall k$, and $\sum_{k=0}^n P_k\leq 1$.
This can be simply achieved (if necessary) by replacing the latter polynomial with $\epsilon p_n(x)$,
 for sufficiently small $\epsilon >0$. The new polynomial has the same set of roots, and it is certainly possible to satisfy all the conditions (\ref{norm}) concerning $k\leq n$. Note again that the relation between the $a_k$'s and the $\Delta_k(0)$'s is linear, and therefore the scaling $p_n\to\epsilon p_n(x)$ results in a scaling $\Delta_k(0)\to \epsilon \Delta_k(0)$. All conditions (\ref{norm}) for $k\leq n$ can be satisfied for an appropriate choice of $\epsilon$, since the number of conditions is finite. As for condition (\ref{norm2}), it can only be violated if $\sum_{k=0}^n \Delta_k(0)>0$, 
since $\sum_{k=0}^n Q_k<1$ is satisfied by construction. In that case, one can just further reduce the value of 
$\epsilon$ above, until (\ref{norm2}) is satisfied. Clearly, this last action preserves conditions~(\ref{norm}), therefore we are done for $k\leq n$. For the remaining initial values one can simply set 
$P_{n+1}(0)=1 -\sum_{k=0}^n\bigl(\Delta_k(0)+Q_n\bigr)$ and $P_k=0$, $\forall k> n+1$. $\square$
\end{proof}

\subsection{Linear Growth Rate}
\label{dif gamma}
When all  values $\lambda_n$ in the sequence $\{\lambda_n=\gamma_n+\mu_n\}$ are different, the general solution for the functions $\Delta_n$ takes the form  \cite{3}:
\begin{equation}
\label{t1}
\Delta_n(t)=\sum_{k=0}^{n} C_k \alpha^n_k e^{-\lambda_k t},
\end{equation}
with
\begin{equation}
\label{t11}
 \alpha^n_k = \frac{\prod_{m=k}^{n-1}\mu_m}{\prod_{m=k+1}^n(\lambda_m-\lambda_k)}.
\end{equation}

The coefficients $C_k$ and the initial conditions $\Delta_k(0)$ are related by: 
\begin{equation}
\label{t12}
\Delta_n(0)=\sum_{k=0}^{n} C_k \alpha^n_k,
\end{equation}
or by the inverse relation:
\begin{equation}
\label{t2}
C_k=\sum_{j=0}^{k} \prod_{i=j}^{k-1} 
\frac{\mu_i}{\lambda_i-\lambda_k}\Delta_j(0).
\end{equation}

The above expressions take a simpler form in the  particular case considered in Ref.~\cite{3}.
Let then:
\begin{equation}
\label{t3}
\lambda_n=\gamma+\sigma(n+1), \quad \sigma> 0, \quad\forall n\geq 0.
\end{equation}

It follows from  (\ref{t11}) that the coefficients $\alpha_k^n$ are just the binomial coefficients, i.e.,
\begin{equation}
\label{t13}
\alpha_k^n=\begin{pmatrix} 
n \\
 k \\
\end{pmatrix}.
\end{equation}

In addition, expression (\ref{t1}) simplifies to:
\begin{equation}
\label{t4}
\Delta_n(t)=e^{-(\gamma+\sigma)t}\sum_{k=0}^{n} C_k 
\alpha_k^n
(e^{-\sigma t})^k,
\end{equation}
with $\alpha_k^n$ given by (\ref{t13}).
Relation (\ref{t2}) in turn becomes: 
\begin{equation}
\label{t5}
C_k= 
\sum_{j=0}^{k} (-1)^{k-j} 
\begin{pmatrix} 
k \\
 j \\
\end{pmatrix}\Delta_j(0).
\end{equation}

Therefore, again $\Delta_n(t)$ is of the form:
\begin{equation}
\label{t6}
\Delta_n= e^{-\lambda t}g_n,
\end{equation}
with $\lambda=\mu+\gamma$,  where $g_n$ is a polynomial, this time in the variable $y=e^{-\sigma t}$, i.e.,
\begin{equation}
\label{t}
g_n=\sum_{k=0}^{n} C_{k}\alpha_k^n y^k, \qquad y=e^{-\sigma t}.
\end{equation}
\label{t7}

The derivative is:
\begin{equation}
\dot{\Delta}_n=e^{-\lambda t}(\dot{g}_n-\lambda g_n),
\end{equation}
and again the stationary points of $\Delta_n$ correspond to the positive roots of the polynomial 
$\dot{g}_n-\lambda g_n$, which now takes the form:
\begin{equation}
\label{t8}
\dot{g}_n-\lambda g_n=-\sum_{k=0}^{n} C_{k}\alpha_k^n (\lambda+k \sigma) y^k.
\end{equation}

Obviously, the stationary points $t$ of $\Delta_n$ are related to the roots $y$ of $\dot{g}_n-\lambda g_n$ by 
$t=-\frac{1}{\sigma}\ln(y)$.

The following analogue of proposition (\ref{cm}) can be proven, using essentially the same arguments, with the appropriate adaptations.
\begin{Proposition}\label{cm2}
Let $\gamma$, $\sigma$, and $n$ be given, and let $t_n$ denote the largest (positive) stationary point of  $\Delta_n$ 
(\ref{t4}).
Then, for any $M>1$, one can find initial values $\Delta_k(0)$, $\forall k\geq 0$, such that $\Delta_n$ has $n$ positive stationary points and $t_n>M$.
\end{Proposition}
\begin{proof}
Let $p_n(y)=\sum_{k=0}^n a_k y^n$ denote an arbitrary degree-$n$ polynomial. Again, the equation 
$\dot{g}_n-\lambda g_n=p_n$ gives us a system of linear equations for the set $\Delta_0(0),\cdots, \Delta_n(0)$,
in terms of the set $a_0,\cdots,a_n$, which can always be solved, taking into account (\ref{t12}).
Now choose $p_n$ such that it possesses $n$  roots in the interval $]0,1]$. This is certainly possible, since one can pick any polynomial $p_n(y)$ with $n$ positive roots and replace it by
$p_n(y/\Lambda)$, for appropriate $\Lambda < 1$. Let $y_0$ be the smallest root of $p_n$. If $t_n=-\frac{1}{\sigma}\ln(y_0)$ is not greater than the prescribed value $M$, then simply adjust the value of $\Lambda$ above. The choice 
$\dot{g}_n-\lambda g_n=p_n$ thus gives us $n$-positive stationary points and an arbitrarily large value of $t_n$. To conclude the proof, it remains to show that the initial values $\Delta_k(0)$ can be chosen such that $P_k=\Delta_k(0)+Q_k$
defines a normalized probability distribution, 
which can be done by exactly the same arguments as in proposition (\ref{cm}). $\square$
\end{proof}

\section{Conclusions and Discussion}
\label{dis}

%%%%%%%%%%%%%%%%%%%%%%%%%%%%%%%%%%%%%%%

We have analyzed the  Biró--Néda model \cite{1,2}, parametrized by the reset rate  $\gamma_n$ and growth rate $\mu_n$. We have shown that, contrary to expectations,   the model does not actually admit (nontrivial) stationary solutions,  if the sequence of values $\mu_n/\gamma_n$ grows faster than $n$. This follows from the fact that  the sum $S_0$ discussed in Ref. \cite{1} fails to be equal to one in such cases. Interesting time-independent distributions are nevertheless associated with the sequences $\mu_n$ and $\gamma_n$, however they simply do not appear as solutions of the model. This is related to the fact that those distributions 
$Q_n$ are such that 
%, contrary to the claims in Ref. \cite{1}.
$\lim_{n\to\infty} \mu_n Q_n$ is nonzero, introducing an unbalanced boundary term at infinity.

Keeping the spirit of the Biró--Néda model, and trying in particular to remain within the setting of unidirectional processes, we have proposed a modified version of the dynamical equation for $\dot{P}_0$ 
(\ref{1}), allowing for those distributions to emerge as true stationary solutions. This effort is partly successful, in the sense that the solutions are included in the model.  However, probability is only conserved close to stationarity. This consequence looks unavoidable,   since the lack of backwards interactions seems to leave no other option  but to include the solutions essentially by introducing a nonhomogeneous term   in the evolution of $P_0$, which is consequently reflected as a boundary term in the equation for the conservation of probability. This is to be contrasted e.g., with birth and death processes, where exactly the same stationary solutions can be obtained, with no boundary issues and with unrestricted probability conservation.

It is perhaps not surprising that technical challenges are arising in the Biró--Néda model for fast growth.
In fact, at least for constant $\gamma_n$, the large $n$ behavior is expected to be dominated by the $\mu_n$ terms in Equation (\ref{0}). This is essentially a pure birth process, and those are known to exhibit exotic behavior for a fast growth rate, e.g.,  {\em explosions} (see e.g., \cite{K,A}) and probability loss \cite{F}, appearing under essentially the same mathematical conditions found in 
Proposition \ref{s}.

Concerning other technical points,   one needs to ensure that $\{Q_n\}$ is normalizable (which we have also shown is not 
{a priori} guaranteed for all sequences $\gamma_n$ and $\mu_n$) and that the infinite sum 
$\sum^{\infty}\gamma_n P_n$ appearing in the differential equation for $P_0$  makes sense. In fact,  the coupling of 
$P_0$ with an infinite number of variables $P_n$   seems technically challenging, and a fast decaying  sequence 
$\gamma_n$ would be preferable for a more amenable formulation. This however conflicts with the normalizability of $\{Q_n\}$.
Requiring $\gamma_n\in ]a,b[$ for all $n$, for some strictly positive $a$ and $a<b<\infty$,  might be sufficient to address both questions, however a more detailed study is probably required.

Finally, in Section \ref{trans}, we  focused on the time-dependent behavior of the functions 
$\Delta_n=P_n-Q_n$, in the amenable cases of a constant reset rate and constant or linear growth rate.  Our analysis considerably expands the discussion of the transient behavior presented in Section 6 of Ref. \cite{3}. The behavior seen in Ref. \cite{3} is not generic, and in fact comes from the very particular initial conditions $P_n(0)=0$, for all 
$n>0$,  which, taking into account  Equation (\ref{0}), effectively compel all but one of the stationary points of $P_n$ to be located at $t=0$.
In general, we were able to show that $\Delta_n$, and therefore $P_n$, can possess $n$ stationary points, the last one occurring at an arbitrarily long time. Thus,  monotonic convergence can occur, for general initial conditions, only after an arbitrarily long time.
%MDPI: Please confirm if it is a citation in the reference.

\vspace{6pt}

\acknowledgments J.V. is grateful for the support given by the research unit Fiber Materials and Environmental Technologies (FibEnTech-UBI), on the extent of the project reference UIDB/00195/2020, funded by the Fundação para a Ciência e a Tecnologia (FCT). I.I. was partially supported by Fundação para a Ciência e Tecnologia through the project UID/MAT/00212/2019.

%\newpage


\begin{thebibliography}{999.}

\bibitem{1} Biró, T.S.; Néda, Z. Dynamical stationarity as a result of sustained random growth. \emph{Phys. Rev. E} 
{\bf 2017}, \emph{95}, 032130.
% https://doi.org/10.1103/PhysRevE.95.032130.

\bibitem{2} Biró, T.S.; Néda, Z. Unidirectional random growth with resetting. \emph{Phys. Stat. Mech. Appl.} {\bf 2018}, \emph{499}, 335--361. 


\bibitem{face} Néda, Z.; Varga, L.; Biró, T.S. Science and Facebook: The same popularity law! 
\emph{PLoS ONE} {\bf 2017},  \emph{12}, e0179656. 

\bibitem{8}  Biró, T.S.; Néda, Z.; Telcs, A. Entropic Divergence and Entropy Related to Nonlinear Master Equations. \emph{Entropy} {\bf 2019}, \emph{21}, 993.
% https://doi.org/10.3390/e21100993.



\bibitem{4} Néda, Z; Gere, I.;  Biró, T.S.; Tóth, G.; Derzsy, N.
Scaling in income inequalities and its dynamical origin.
\emph{Phys. Stat. Mech. Appl.} {\bf 2020},  \emph{549}, 124491. 
% https://doi.org/10.1016/j.physa.2020.124491

\bibitem{5} Gere, I.; Kelemen, S.;  Toth, G.; Biró, T.S.;  Neda, Z. Wealth distribution in modern societies: collected data and a master equation approach.
\emph{Phys. Stat. Mech. Appl.} {\bf 2021}, \emph{581}, 126194.
% https://doi.org/10.1016/j.physa.2021.126194.

\bibitem{6}   Gere, I.;  Kelemen, S.;  Biró, T.S.;  Néda, Z. Wealth distribution in villages. Transition from socialism to capitalism in view of exhaustive wealth data and a master equation approach. \emph{Front. Phys.} {\bf 2022}, \emph{10}, DOI:10.3389/fphy.2022.827143.
% https://doi.org/10.3389/fphy.2022.827143.


\bibitem{w1} Park, J.; Park, Y.  Wealth Distribution for the Spin Agent Model of the Stock Market.  \emph{NPSM}  {\bf 2020}, \emph{70}, 292--298.  
 
\bibitem{w2} Cardoso, B.-H.;  Gonçalves, S.;  Iglesias, J.
Wealth distribution models with regulations: Dynamics and equilibria.
\emph{Phys. Stat. Mech. Appl.} {\bf 2020}, \emph{551},
124201.
% https://doi.org/10.1016/j.physa.2020.124201.


\bibitem{w3} Cui, L.;  Lin, C.
A simple and efficient kinetic model for wealth distribution with saving propensity effect: Based on lattice gas automaton.
\emph{Phys. Stat. Mech. Appl.} {\bf 2021}, \emph{561},
125283.
% https://doi.org/10.1016/j.physa.2020.125283.
 


\bibitem{VK} Van Kampen, N.G. \emph{Stochastic Processes in Physics and Chemistry}; Elsevier Science B.V.: Amsterdam, The Netherlands, 1992.

\bibitem{R} Ross, S.M. \emph{Introduction to Probability Models}; Academic Press: Oxford, UK, 2019.

\bibitem{War}  Irwin, J.O. The generalized Waring distribution applied to accident theory. 
\emph{J. Roy. Stat. Soc. A}  {\bf 1968}, \emph{131}, 205--225.
% https://doi.org/10.2307/2343842.


\bibitem{cor}  Corominas-Murtra, B;  Hanel, R.;  Zavojanni, L.;   Thurner, S. How driving rates determine the statistics of driven non-equilibrium systems with stationary distributions.
 \emph{Sci. Rep.} {\bf 2018},  \emph{8},  10837.
% https://doi.org/10.1038/s41598-018-28962-1.


\bibitem{3} Biró, T.S.; Csillag, L.; Néda, Z. Transient dynamics in the random
growth and reset model. \emph{Entropy}
{\bf 2021}, \emph{23}, 306.
% https://doi.org/10.3390/e23030306.

\bibitem{K} Kirkwood, J.R. \emph{Markov Processes}; CRC Press: Boca Raton, FL, USA, 2015.

\bibitem{A} Anderson, W.J. \emph{Continuous-Time Markov Chains: An Applications-Oriented Approach}; Springer-Verlag: New York, NY, USA, 1991. 


\bibitem{F} Feller, W. \emph{An Introduction to Probability Theory and Its Applications}; John Wiley nad Sons, Inc.: New York, NY, USA, 1968; Volume 1.






\end{thebibliography}
\end{document}